\documentclass[12pt]{amsart}
\usepackage{color,amsmath,amsthm}

\usepackage[T1]{fontenc}
\usepackage[latin1]{inputenc}
\usepackage{amsmath}
\usepackage{amsfonts}
\usepackage{bm}
\usepackage{graphicx}
\usepackage{subfig}
\usepackage{pst-plot}
\usepackage[hmargin=1in,vmargin=1in]{geometry}
\usepackage{enumerate}
\usepackage{pslatex}
\usepackage{ifthen}
\usepackage{calc}
\usepackage{answers}

\newtheorem{theorem}{Theorem}
\newtheorem{cor}{Corollary}
\newtheorem{rmk}{Remark}

\newcommand{\Z}{{\mathbb Z}}

\author{Ariel Lerman and Vadim Zharnitsky }

\address{ Department of Mathematics \\ University of Illinois \\ Urbana, IL 61801}

\title{Whispering gallery orbits in Sinai oscillator trap}

\begin{document}

\begin{abstract}
Experimental realizations of  trapping Bose-Einstein condensate lead to a Hamiltonian  system of a classical particle bouncing  off a convex scatterer in the field of an attracting potential \cite{Shepelyansky}. It is shown by application of KAM theory that under some natural conditions there exists positive measure of quasiperiodic solutions near the scatterer's boundary. 
\end{abstract}

\maketitle

\section{Introduction} 
Recently, there have been increased interest in the  systems of billiard type in which a classical  particle is trapped by a convex  attracting potential and scattered from the center of attraction   by a convex  scatterer such as a circle or ellipse \cite{Shepelyansky,ermann,frahm}. The main motivation comes from 
the dynamics of Bose-Einstein condensate in which cold atoms are trapped by harmonic potential and strong optical beam creates a repulsive potential which can be approximated  by a perfectly elastic scatterer.

One of the  simplest, yet, nontrivial examples is provided by the Hamiltonian system of a particle in a harmonic non-centrally symmetric potential with a scattering disc centered at the origin, so 
the Hamiltonian function takes the form
\begin{equation} \label{eq:example}
E = \frac{1}{2m}(p_x^2+p_y^2) + \frac{m}{2} (\omega_x^2 x^2 + \omega_y^2 y^2) + V_d(x,y),
\end{equation}
where $V_d$ corresponds to the potential of the repulsive disk.

 There is some analogy with the well known example of a Sinai type billiard in a square with the  disk scatterer. Such billiards are known to be ergodic \cite{sinai} and one might expect that similar behavior could be observed for the Sinai oscillator traps which were introduced  in \cite{Shepelyansky}. Indeed the results of numerical simulations show largely 
chaotic behavior in such systems with some possibility for elliptic islands in the phase space.

The goal of this article is to provide a simple criterion when such systems are definitely not ergodic for some range of energies in the sense that the corresponding fixed   energy level sets have at least two invariant components of positive measure.  At the same time, we establish analogues of Lazutkin's caustics (also known as whispering gallery orbits by analogy with the sound waves propagating near the walls of  medieval cathedrals) near the boundary of the scatterer. Our approach is the local analysis of the solutions near the boundary of the scatterer. We show that, under some conditions, the solutions stay near the boundary for all time, giving rise to invariant curves in the Poincar\'e map. It is our hope that this analysis will be useful in modeling of Bose-Einstein condensate in the presence of steep repulsive potential.

Another  motivation to consider such systems comes from engineering applications in automation and control problems. There has been a lot of interest in the mixed dynamics when there are both discrete and continuous subsystems in a single dynamical system. 
The billiard-in-a-potential system is one of the simplest examples of the so-called hybrid systems. For more background on the hybrid or switched systems, see multiple books and monographs, e.g. \cite{Liberzon}.

The idea to  consider orbits nearly tangential to the boundary of a billiard domain and realization  that they will stay there for a long time, dates back at least to Birkhoff  \cite{birkhoff}.  He considered  an orbit that  is nearly tangent to the boundary of a convex billiard domain and therefore  bounces many times before the curvature undergoes considerable  change. Next, using a contemporary language, separation of scales can be established which leads to an adiabatic invariant: an approximately conserved quantity corresponding to an action variable in the appropriately defined  action-angle variables.
With the creation of KAM theory, it became  possible to prove that under some non-degeneracy conditions  the adiabatic invariant was restricted to some small  range of values  for all time. In  \cite{arnold}, V.I. Arnold proved that under some smoothness and non-degeneracy conditions  a slowly varying oscillatory Hamiltonian system 
\[
H = H(p,q,\lambda = \epsilon t)
\]
has a perpetually conserved adiabatic invariant $J(p(t),q(t),\lambda(t))$. Later, a similar technique was developed for Hamiltonian systems  with low regularity in the context of the  Littlewood problem, see e.g. \cite{Kunze,Levi}. The main idea was to organize  a series of canonical transformations, which on the one hand, bring the system to  the near-integrable form, and on the other hand, make the  "non-smooth" variable into an evolutionary parameter (time). After integrating the corresponding Hamiltonian  ODEs over the new time variable, one obtains a smooth near-integrable map to which Moser's small twist theorem can be applied to establish a large set of quasiperiodic orbits. Subsequently, this approach was extended to the   problems with impacts, including billiards, so one could establish an adiabatic invariant in non-smooth Hamiltonian systems of the type
\[
H = H(|x|,y, \epsilon t),
\]
see e.g.  \cite{vz}.
For more background on billiard problems, see \cite{tabachnikFrance,tabachnikAMS}.

We use this approach to establish a large set of  quasiperiodic solutions near the boundary of the scatterer by applying KAM theory. Next, we illustrate how the theorem applies in several specific systems, such as given by \eqref{eq:example} with a circular scatterer. We also discuss the corresponding adiabatic invariant in the Sinai oscillator trap.


\section{Model and main results}
Consider a classical particle in a convex potential $V(x,y)>0$ if $(x,y) \neq (0,0)$, $V(0,0)=0$. Let  $\Omega$ be a convex domain containing the origin. We assume that the particle bounces from
 the boundary of the domain according to the  usual billiard law: angle of reflection is equal to the angle of incidence.  We always  assume that both the potential and the domain boundary are given by analytic functions although this condition can be relaxed.

 \begin{figure}[ht]
   
    \includegraphics[scale = 0.38]{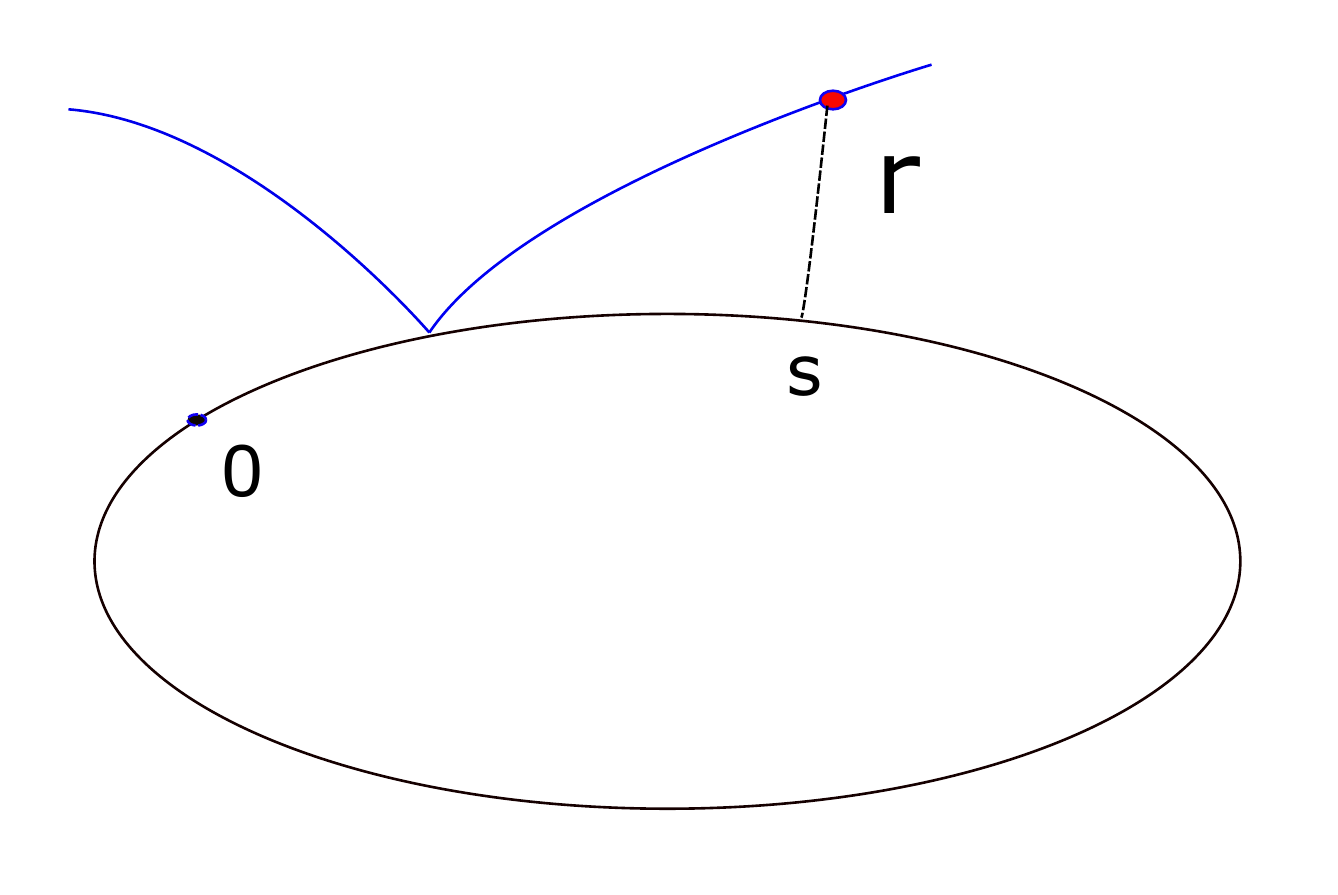}

    \caption{Boundary coordinates. Any point in the exterior of the convex scatterer has a uniquely defined coordinate $r$ equal to the distance from the boundary and coordinate $s$ along the boundary from some marked point $s=0$, measured in the clockwise direction. }

\label{fig:boundary_coord}

\end{figure}
 
 Near the boundary $\partial \Omega$ we introduce the so-called  boundary coordinates $(r,s)$, $r$ is the distance from the boundary and $s$ is the  length along the boundary, see Figure \ref{fig:boundary_coord}.

  \begin{figure}[ht]
    \centering
    \subfloat[]{{\includegraphics[scale = 0.38]{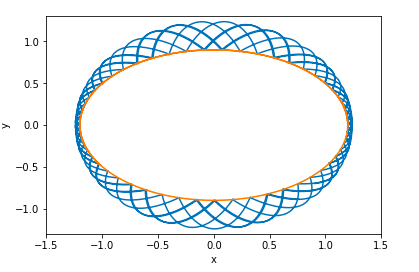} }}%
    \qquad
    \subfloat[]{{\includegraphics[scale = 0.38]{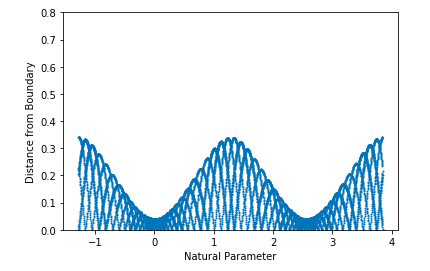} }}%
    \caption{Panel (A):  Orbit which stays near the boundary  in a billiard around an ellipsoidal scatterer in $x$-$y$ coordinates. Panel (B): The same orbit in a billiard around an ellipsoidal scatterer in boundary coordinates}%
    
    \label{fig:circle}

\end{figure}
 
 Assume the potential can be represented in the exterior of the domain $\Omega$ by Taylor series in $r$
 \begin{equation}
 V(r,s) = V_0(s) + r V_1(s) + \frac{r^2}{2}V_2(s) ...
 \end{equation}
 
 The Lagrangian of the classical particle in the potential in the boundary coordinates takes the form
 \begin{equation}
 L = \frac{\dot r^2}{2} + (1+k(s)r)^2 \frac{\dot s^2}{2}  -V(r,s),
 \end{equation}
 which after Legendre transformation with $p_r = \dot r, p_s = \dot s (1+k(s)r)^2$ leads to the  Hamiltonian
  \begin{equation}
 E = \frac{p_r^2}{2} + \frac{p_s^2}{2(1+k(s) r)^2} +V(r,s).
  \end{equation}
 We modify the system by letting $r$ take negative values so we get
   \begin{equation}
 E = \frac{p_r^2}{2} + \frac{p_s^2}{2(1+k(s) |r|)^2} +V(|r|,s).
  \end{equation}
 The last modification is done just for convenience and can be avoided. 
 
 Using the invariant relation  of the 1-form $p_r dr + p_s ds - E dt = p_r dr + E d(-t) - (-p_s) ds$ with the Hamiltonian flow, we obtain an equivalent Hamiltonian system 
  \begin{equation}
 K = -(1 + k(s)|r|) \sqrt{2E-p_r^2 -2V},
  \end{equation}
 where $K= -p_s$ is the new Hamiltonian function and $s$ is the new time. \\
 
 Next, we use the rescaling $K =\epsilon^2 F, s=\epsilon S$, $p_r = \epsilon P_R, r = \epsilon^2 R$ which is convenient for studying orbits in the vicinity of the boundary. The new Hamiltonian takes the form
 \[
 F = -(\epsilon^{-2} + k(\epsilon S) |R|) \sqrt{2 E -2V_0(\epsilon S) -\epsilon^2 P_R^2 -\epsilon^2 2 V_1(\epsilon S)|R| +O(\epsilon^4) }   + O(\epsilon^2).
 \]
 After expanding the square root and some rearrangements, we have
 
  \begin{equation}
 F = \frac{P_R^2}{2a(\epsilon S)} + b(\epsilon S) |R| + \epsilon^2 F_1(|R|, P_R, \epsilon S, \epsilon),
 \end{equation}
 where 
 \[
 a(\epsilon S) = \sqrt{2E - 2V_0(\epsilon S)}
 \]
 and
 \[
 b(\epsilon S) = -k(\epsilon S) \sqrt{2 E -2V_0(\epsilon S)} + \frac{V_1(\epsilon S)}{  \sqrt{2 E -2V_0(\epsilon S)}}.
 \]
 Note that we have dropped the term $\epsilon^{-2}$ as it does not affect the Hamiltonian flow. On the other hand, because of this term $F$ is of order $\epsilon^{-2}$  and we have that  $p_s = -K= -\epsilon^2 F$ is of order 1. 
 In order to have an oscillatory Hamiltonian system, we must require $a>0$ which would follow if $V_0(\epsilon S) <  E$ and we also need $b>0$. 
 The latter condition holds if the potential is strong enough near the boundary to counteract the curvature effects
 \[
 V_1(\epsilon S) > k(\epsilon S ) (2 E-2V_0(\epsilon S)).
 \]
 
As the last step, we introduce more convenient variables $R=x, P_R =y, S=t, F=H$ and we obtain explicitly 
a slowly varying oscillatory  non-smooth Hamiltonian system 
 \begin{equation}\label{eq:basic}
 H = \frac{y^2}{2a(\lambda)} + b(\lambda) |x| + \epsilon H_1(|x|,y,\lambda,\epsilon),
 \end{equation}
 where $\lambda =\epsilon t$.
 This system is in the form for which a non-smooth version of  adiabatic invariance theorem can be proved

 \begin{theorem}
 Let $k(s), V_0(s), V_1(s),E$ satisfy  $E > V_0(s)$, $V_1(s) > 2 k(s) (E-V_0(s))>0$ for all $s\in [0,L)$, then there exist $\epsilon_0>0$ and $C>0$ such that for any $0< \epsilon < \epsilon_0$, if $|r(0)| < \epsilon^2, |\frac{dr}{ds}(0)| < \epsilon$, the solution will stay near the boundary
 $|r(s)| < C \epsilon^2 \,$ for all time.
 \end{theorem}
 
 \begin{proof}
 The main idea of the proof is the reduction of the Hamiltonian system in the form \eqref{eq:basic} to the near-integrable form so a version of KAM theory could be applied.  \\
 
 \noindent
 {\bf Canonical transformations.} 
  First, we carry out the standard  transformation  to the action-angle variables, where the action variable is determined by the leading part of the Hamiltonian \\
 $H_0 = \frac{y^2}{2a(\lambda)} + b(\lambda) |x|$ as the area enclosed by the phase curve, so then 
 \begin{equation}\label{eq:action_0}
 I = 4 \int_{y_{-}}^{y_{+}} x(H,y) dy = 4 \int_0^{\sqrt{2aH}}  \left (  H-\frac{y^2}{2a} \right ) \frac 1 b \,dy = \frac 8 {3b(\lambda)} \sqrt{2a(\lambda)} H^{3/2}.
 \end{equation}
 The angle variable  then can be defined on the interval $[-1/2,1/2)$, the Hamiltonian function is 1-periodic in $\phi$ and depends on the absolute value of $\phi$.

The variables change according to 
 \[
 (x,y,t,H(|x|,y,t)) \rightarrow (\phi, I, t, H(|\phi|,I,t)).
 \]
  We use the ordering: (position, momentum, evolutionary parameter (time), Hamiltonian function), so the new position variable is $\phi$, new momentum is $I$, etc.

 The Hamiltonian function in the action-angle variables takes the form
 \[
 H = H_0(I,\lambda) + \epsilon H_1(I, |\phi|, \lambda, \epsilon),
 \]
 where the smallness of the  perturbation $H_1$ is due to the slow time dependence
  
 Our second step, is to switch the role of $(\phi,I)$ and $(\lambda = \epsilon t , H)$:
 \[
 (\phi, I, t, H(|\phi|,I,t)) \rightarrow (t,H, \phi, I(t,H,|\phi|)), 
 \]
 using invariance of the Poincar\'e form with respect to the  Hamiltonian flow
 \[
 Id\phi -H(I, |\phi|, \lambda, \epsilon) dt = -\frac{1}{\epsilon} (Hd\lambda - \epsilon I(H,\lambda ,|\phi|,  \epsilon) d \phi).
 \]
 The Hamiltonian function takes the form
 \[
\tilde  I = \epsilon I_0(H,\lambda) + \epsilon^2  I_1 (H, \lambda,  |\phi|, \epsilon),
 \]
 where $I_0$ is given by \eqref{eq:action_0}. As above, the Hamiltonian function is analytic in all variables away from the set $\phi =0, 1/2$. \\
 
The last third  step is to apply an analogue of the action-angle transformation to $I_0(H, \lambda)$, so we have the following change of variables 
\[
 (\lambda ,H, \phi, I(\lambda ,H,|\phi|)) \rightarrow (\tau, h, \phi, J(\tau,h,|\phi| )  ).
\]
The Hamiltonian takes the form 
\[
J(\tau,h,|\phi| ) =  \epsilon J_0(h)+ \epsilon^2  J_1 (\tau, h, |\phi|) .
\]

To obtain explicitly the last transformation, we first rewrite \eqref{eq:action_0} as
\[
I = \sigma(\lambda) H^{\frac 3 2},
\]
where $\sigma(\lambda +L) = \sigma(\lambda)$ is $L-$periodic with $L$ being the length of the boundary of the domain. \\

We look for a symplectic transformation of the form
\begin{eqnarray}
h &=& c_1 H \sigma^{2/3}(\lambda) \\
\tau &=& c_2 \int_0^{\lambda} \sigma^{-2/3} (\mu) d\mu. \nonumber
\end{eqnarray} 
To obtain $dH\wedge d\lambda = dh\wedge d\tau$, we need $c_1c_2 =1$. The freedom in the choice of these constants is used to normalize $\tau$, so that it has period 1. Thus, 
\[
c_1 = \frac{1}{c_2} =  \int_0^{L} \sigma^{-2/3} (\mu) d\mu
\]
and then, 
\[
J_0(h) =  \sigma(\lambda) H^{\frac 3 2} = \sigma(\lambda) \left (\frac{h}{c_1 \sigma(\lambda)^{2/3}} \right )^{3/2} = \frac{h^{3/2}}{ ( \int_0^L \sigma^{-2/3}(\mu) d\mu )^{3/2}}.
\]

 The last transformation leads to the Hamiltonian function which is explicitly a small perturbation of an integrable one $J_0(h)$.  As before, all  non-smooth dependence is contained in the evolutionary parameter $\phi$ 
 at $\phi = 0, 1/2$. 
 Integrating the Hamiltonian equations from $\phi=0$ to $\phi = 1/2$, we obtain the billiard (collision) map which is an analytic  map 
  \begin{eqnarray}\label{eq:mainmap}
 \tau_2 &=& \tau_1 +  \epsilon \, J_0'(h_1) +  \epsilon^2 f(\tau_1, h_1, \epsilon) \\
 h_2 &=& h_1 + \epsilon^2 g(\tau_1, h_1, \epsilon) , \nonumber
 \end{eqnarray}

  to which Moser's small twist theorem can be applied. \\

 \noindent
 {\bf Moser's small twist theorem.} For the reader's convenience, we recall the statement of the Moser's twist theorem, see   \cite{Moser}. We use the original version since we are not concerned with the optimal smoothness conditions. For more refined versions, see e.g. \cite{LeviMoser} and references therein. 
 
 Assume that the area-preserving smooth map $\Psi$  of the set $D = \{ x\in {\mathbb R}, a < y < b\}$, depending on a parameter $\gamma>0$, has the form
 \begin{eqnarray}
 x_2 &=& x_1 + \gamma \, y_1 + \gamma F(x_1, y_1, \gamma) \, \label{eq:mtm} \\
 y_2 &=& y_1 + \gamma \, G(x_1, y_1, \gamma) , \nonumber
 \end{eqnarray}
 where $F$ and $G$ are differentiable sufficiently many times and 1-periodic in $x$. 
 Assume that the map possesses the curve intersection property: any closed curve given by a graph $y = f(x)$ with $f'$ sufficiently small, will intersect its image.\footnote{If the map is homeomorphic to  an area-preserving map, then the self-intersection property clearly holds.}
 
 We will denote by $||F||_{k}$ the usual norm in the space of  functions which have $k$ continuous derivatives
 \[
 ||F||_k = \sup_D  \,\, \left |  \left (\frac{\partial}{\partial x} \right )^{\kappa_1}  \left (\frac{\partial}{\partial y} \right )^{\kappa_2} F(x,y) \right |, \kappa_1+\kappa_2 \leq k.
 \]
 

 \begin{theorem}[Moser's small twist theorem]
 For any $\mu > 0$ and integer $s\geq 1$, there exists $\delta(c_0,\mu,s) > 0$, independent of $\gamma$, and integer $l(s)$ such that  if
 \[
 ||F||_{0} + ||G||_{0} <   \, \delta, \,\,\, ||F||_{l} + ||G||_{l} <  c_0
 \]
 then there exists an invariant curve in $D$, given by  $x =  \theta + u(\theta), y= y_0+v(\theta)$  with both $u(\theta)$ and $v(\theta)$ of period 1, $||u||_s + ||v||_s < \mu$. The mapping induced on the invariant curve is given by $\theta_2 = \theta_1 + \gamma y_0$ with $a+\mu < y_0 < b - \mu$.
 \end{theorem}
 
 To apply the small twist  theorem, we make the transformation $y = J'(h), x=\tau, \gamma = \epsilon $ in \eqref{eq:mainmap} using that $J'(h) \neq 0$. The transformed map is in the same form as \eqref{eq:mtm}.
 The map \eqref{eq:mainmap} is analytic by our assumptions and therefore the bound on higher derivatives  $||\epsilon f||_l + ||\epsilon g||_l < c_0$ holds for sufficiently small $\epsilon$, as well as 
 $||\epsilon f||_0 + ||\epsilon g||_0 < \delta$, taking $\epsilon$ even smaller if necessary.  Thus, we obtain an invariant curve in the billiard map which gives rise to  a quasiperiodic solution in the full system.
 
  \end{proof}

 \begin{cor}
 By a standard argument in KAM theory, as the small parameter  $\epsilon \rightarrow 0$, the relative measure of the quasiperiodic solutions tends to the full measure. Therefore in the $\delta$ neighborhood of the boundary the relative 
 measure of the set occupied by quasiperiodic solutions tends to 1 as $\delta \rightarrow 0$.
 \end{cor}

 \begin{cor}
The above system  possesses a perpetually conserved adiabatic invariant (an approximately conserved quantity in a slowly varying Hamiltonian system.). It is given by 
$h = c_1 H \sigma^{2/3}(\lambda)$ according to the theorem. Since adiabatic invariant is defined up to a multiplication by a constant and there is a correspondence $H \rightarrow p_s \rightarrow \dot s$
 we conclude that in the billiard map 
 $I =  |\dot s|  \, \sigma^{2/3}(s)$ is 
conserved with the accuracy of order $O(\epsilon)$.
 \end{cor}
  
 \section{Applications}
 
 In this section, we  consider three examples of potentials and scatterers and apply the theorem  from the previous section to establish positive measure of quasiperiodic solutions, which in turn implies the lack of ergodicity.

 \subsection{Circular scatterer, centered at the origin with non-symmetric harmonic potential}
Let $a$ be the radius of the circle and let $V(x, y) = \omega_x^2 x^2 + \omega_y^2 y^2$. We have 
\begin{theorem}
 Let the  billiard system of a particle in a harmonic potential with circular scatterer of radius $a>0$ be such that 
 $$
2\min(\omega_x^2, \omega_y^2) > \max(\omega_x^2, \omega_y^2).
$$
 If the energy $E$ satisfies   $$ \max(\omega_x^2, \omega_y^2)  < \frac{E}{a^2} <  2\min(\omega_x^2, \omega_y^2), $$ then there is a positive measure of quasiperiodic orbits near the boundary.
  
\end{theorem}
\begin{proof}

Note that $x = (r + a)\cos \frac{s}{a}$, and $y = (r + a)\sin \frac{s}{a}$. Hence
$$
V(r, s) = \omega_x^2 (r + a)^2 \cos^2 \frac{s}{a} + \omega_y^2 (r + a)^2 \sin^2 \frac{s}{a}.
$$
Defining $r = \epsilon^2 R$, $s = \epsilon S$, we get that 
\begin{align*}
V(R, s) = \omega_x^2 a^2 \cos^2\frac{\epsilon S}{a} &+ \omega_y^2 a^2 \sin^2 \frac{\epsilon S}{a}\\
&+ 2\epsilon^2 \omega_x^2 R a \cos^2 \frac{\epsilon S}{a} + 2\epsilon^2 \omega_y^2 R a \sin^2 \frac{\epsilon S}{a}\\
&+ \mathcal O(\epsilon^4).
\end{align*}
Therefore, we have  $$V_0(\epsilon S) = \omega_x^2 a^2 \cos^2\frac{\epsilon S}{a} + \omega_y^2 a^2 \sin^2 \frac{\epsilon S}{a} ,$$
$$
V_1(\epsilon S) = 2 \omega_x^2 a \cos^2 \frac{\epsilon S}{a} + 2\omega_y^2 a \sin^2 \frac{\epsilon S}{a}.
$$
Now, we find conditions for $E > V_0$ and $V_1 > \frac{1}{a}(2E - 2V_0)$. Let $\lambda = \frac{\epsilon S}{a}$ and note that $V_0$ and $V_1$ have critical points at $\lambda = n \frac{\pi}{2}$, $n \in \Z$. We also have that $V_0(n \pi) = \omega_x^2 a^2$ and $V_0((2n + 1)\frac{\pi}{2}) = \omega_y^2 a^2$. Thus, we obtain the lower bound on the energy  $E > \max(\omega_x^2, \omega_y^2) a^2$.

Note that $V_1 - \frac{1}{a}(2E - 2V_0)$ also has critical points at $\lambda = n \frac{\pi}{2}$ since $E$ is constant.
At $\lambda = n \pi$, $$V_1 - \frac{1}{a}(2E - 2V_0) = 2 \omega_x^2 a - 2\frac{1}{a}(E - \omega_x^2 a^2) = 4  a \omega_x^2 - 2\frac{E}{a},$$ and at $\lambda = \frac{(2n + 1)\pi}{2}$,
$$
V_1 - \frac{1}{a}(2E - 2V_0) = 4 a \omega_y^2 - 2\frac{E}{a}.
$$
Thus $$a^2 > \frac{E}{2\min(\omega_x^2, \omega_y^2)},$$
and hence we must have 
\begin{equation}\label{eq:cond_freq}
2 a^2 \min(\omega_x^2, \omega_y^2) > E > a^2 \max(\omega_x^2, \omega_y^2).
\end{equation}
\end{proof}

%
%

Now, we present the results of numerical simulations of such system, see Figure \ref{fig:circle_2}. We used circular scatterer of radius 1 and harmonic potential with $\omega_x = 1.25, \omega_y = 1$, so that they satisfy the condition :
\eqref{eq:cond_freq}

\begin{figure}[ht]
    \centering
    \subfloat[]{{\includegraphics[scale = 0.38]{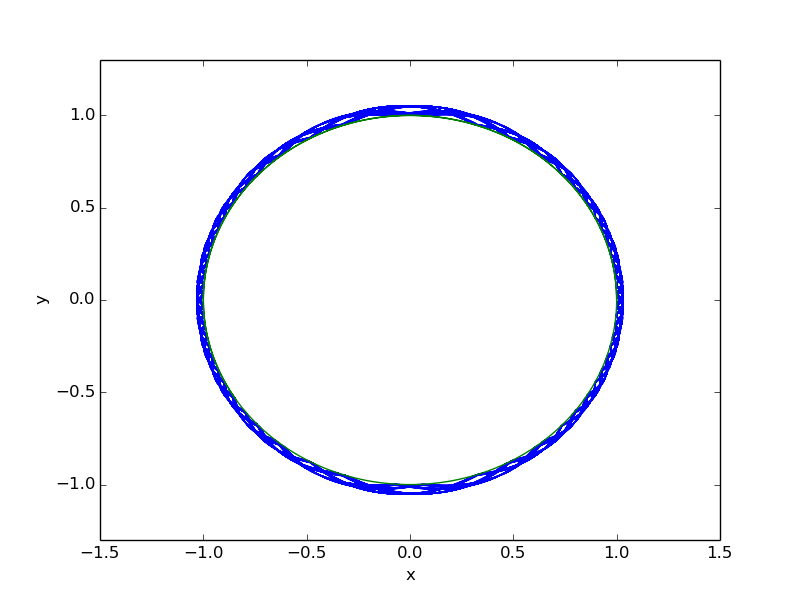} }}%
    \qquad
    \subfloat[]{{\includegraphics[scale = 0.38]{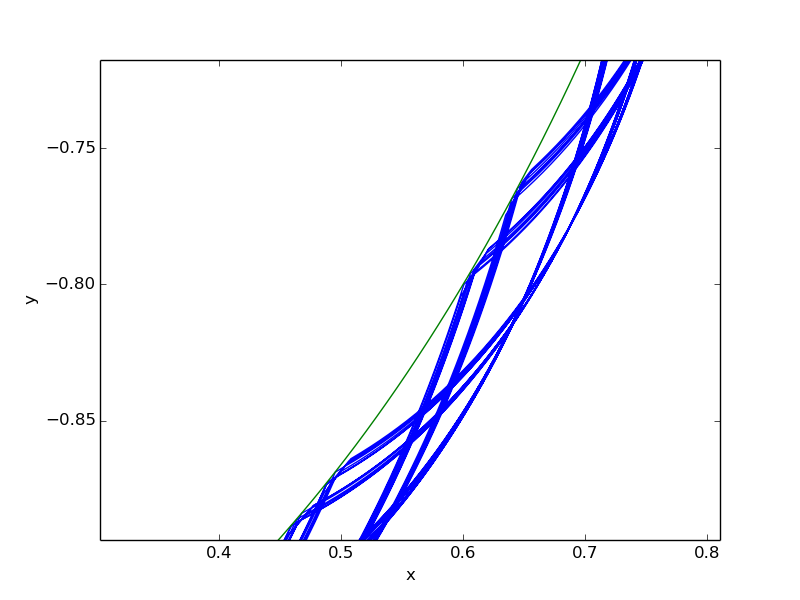} }}%
    \caption{Panel (A):  Orbit with a circular scatterer of radius 1 and harmonic potential, $\omega_x = 1.25, \omega_y = 1$. Initial conditions: $x_0 = 1.03$, $y_0 =0$, $\dot x_0 = 0,$ $\dot y_0 = 0.4,$ $T =  5$ sec.  Panel (B): Magnification near a point at the scatterer.  }%
    
    \label{fig:circle_2}

\end{figure}

 \subsection{ Off-center circular scatterer in linear potential}
 Let $a$ be the radius of the circle, and $c$ be the offset along the x-axis. Assume $c, a > 0$, then we have 
 \begin{theorem}
 Let the billiard system of a particle in linear potential $V = \mu (x^2 + y^2)^\frac{1}{2}$ with circular off-center scatterer be such that $a^2 > 4c^2 + 5 ac$.
  If 
  \[
  \mu(c+a) < E <  \mu \frac a 2 \frac{a-c}{a+c} + \mu(a-c), 
  \]
  then there is a positive measure of quasiperiodic orbits near the boundary.
 
 \end{theorem}
 
 \begin{proof}
In the case where $V = \mu(x^2 + y^2)^\frac{1}{2}$, we have that in the new boundary  coordinates, 
\begin{equation}
V = \mu ((r + a)^2 + 2(r + a)c \cos \frac{s}{a} + c^2)^{\frac{1}{2}}.
\end{equation}
Thus for $r = \epsilon^2 R$, $s  = \epsilon S$,
$$
V = \mu (c^2 + (2 c a \cos \frac{\epsilon S}{a} ) + a^2)^\frac{1}{2} + \frac{\mu}{2} \epsilon^2 (2 c R \cos \frac{\epsilon S}{a} + 2 a R) (c^2 + (2 c a \cos \frac{\epsilon S}{a} ) + a^2)^{-\frac{1}{2}} + o(\epsilon^2)
$$
and therefore we have  $$V_0 = \mu (c^2 + 2 c a \cos\frac{\epsilon S}{a} + a^2)^{\frac{1}{2}}$$ 
and
$$
V_1 = \frac{\mu}{2} (2 c  \cos\frac{\epsilon S}{a} + 2a)(c^2 + 2ca \cos\frac{\epsilon S}{a} + a^2)^{-\frac{1}{2}}.
$$
We will require  $a>c$ for otherwise the linear part of the potential $V_1$ is not well defined for some values of $s$.

Thus both $V_0$ and $V_1$ have critical points at $\lambda = n \pi$, where $\lambda = \frac{\epsilon S}{a}$. 
To get $E > V_0$, we must  have $$E> V_0(n \pi) = \mu |c \pm a|.$$ Thus $E > \mu (c + a)$, since $c, a > 0$.\\

To satisfy the condition $V_1 - \frac{1}{a} (2E - 2V_0) > 0$, we rewrite it as 
\[
\frac{aV_1}{2} + V_0 > E.
\]
Minimizing the left hand-side, we obtain
\[
\frac a 2 \mu \frac{a-c}{a+c} + \mu(a-c) > E > \mu(a+c).
\]
To have nonempty range of $E$, we need to satisfy the inequality
\[
a^2 > 4c^2 + 5 ac,
\]
which also implies $a>c$.


\end{proof}


 \subsection{ Ellipsoidal scatterer with linear potential.}

 \begin{theorem}
Consider  the billiard system of a particle in linear potential $V = \mu(x^2 + y^2)^\frac{1}{2}$ and the ellipsoidal scatterer 
given by
  \[
 \frac{x^2}{a^2}+ \frac{y^2}{b^2} = 1. 
 \]
 There is a positive measure of quasiperiodic orbits near the boundary if 
 $E$ satisfies the inequalities
 
\[
\max \{a,b\} <\frac{E}{\mu} < \min \left  \{ \frac{2a^2+b^2}{2a}, \frac{2b^2+a^2}{2b} \right \}.
\]
 
 \end{theorem}
 \begin{proof}

 We want to see under what conditions on $a,b, \mu$ the theorem applies. Besides smoothness assumptions, which are clearly satisfied, we need
 \[
   V_1 >  2k (E -V_0) >0,
 \]
 where $k(s)>0$ for any ellipse.

 \begin{figure}[ht]
    \centering
    \subfloat[]{{\includegraphics[scale = 0.5]{Figure_3.png} }}%
    \qquad
    \subfloat[]{{\includegraphics[scale = 0.5]{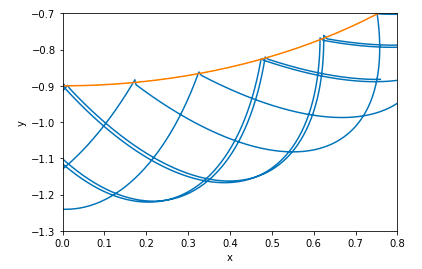} }}%
    \caption{Panel (A):  Orbit with an ellipsoidal scatterer of major axis $1.2$ and minor axis $0.9$, with and linear potential, $\mu = 1$. Initial conditions: $x_0 = 0.9$, $y_0 =-0.7$, $\dot x_0 = 0.6,$ $\dot y_0 = 0.1$.  Panel (B): Magnification near a point at the scatterer.  }%
    
    \label{fig:circle_3}

\end{figure}

 We can rewrite these inequalities as
 \[
 \frac{1}{2k} V_1 + V_0 > E  > V_0,
 \]
 which should hold at each point on the boundary.
 
 We have
 \[
 V_1 = \left .  \frac{\partial V}{\partial r} \right |_{r=0}  = \nabla V \cdot {\bf n} = \mu \cos(\alpha),
 \]
  
 where $\alpha$ is the angle between the  vector in the radial direction and the normal vector to the ellipse boundary pointing outside. 
 
 This leads to a geometric condition 
 \begin{equation}\label{eq:2ineq}
 \frac{\mu \cos(\alpha(s))}{2k(s)} + \mu\sqrt{x^2(s)+ y^2(s)} > E >  \mu\sqrt{x^2(s)+ y^2(s)},
 \end{equation}
 where $x(s), y(s)$ is any parameterization of ellipse (not necessarily a natural one).

We have that $$\cos(\alpha(s)) = \frac{(x(s), y(s)) \cdot {\bf n}}{\sqrt{x^2(s) + y^2(s)}||{\bf n}||}.$$ Thus using the parametrization $(x(s), y(s)) = (a\cos s, b\sin s)$, we obtain  $${\bf n} =  \frac{(b\cos s, a \sin s)}{\sqrt{b^2\cos^2 s+ a^2 \sin^2 s}}$$

and
$$
\cos(\alpha(s)) = \frac{  ab}{\sqrt{x^2(s) + y^2(x)}\sqrt{b^2\cos^2(s) + a^2\sin^2(s)}}.
$$
The curvature is given by
\[
k = \frac{|x''y'-x'y''|}{(x^{'2}+y^{'2})^{3/2}} = \frac{ab}{(a^2 \sin^2s +b^2 \cos^2 s)^{3/2}}.
\]
Thus,
\[
g(s)= \frac{ \cos(\alpha(s))}{2k(s)} + \sqrt{x^2+ y^2} = \frac{a^2 \sin^2 s + b^2 \cos ^2 s }{2\sqrt{x^2+y^2}} + \sqrt{x^2+y^2} = \frac{a^2 +b^2 + a^2 \cos^2 s + b^2 \sin^2 s}{2 \sqrt{a^2 \cos^2 s + b^2 \sin^2 s}}.
\]
Critical points of this expression are $s=0, \pi/2, \pi, 3\pi/2$ and by symmetry we only need to evaluate at $s=0,\pi/2$, which provide 

\[
g(0) = \frac{2a^2+b^2}{2a} , \,\,\,  g(\pi/2) = \frac{2b^2+a^2}{2b}.
\]
Thus, from \eqref{eq:2ineq}, we obtain

\[
\max \{a,b\} <\frac{E}{\mu} < \min \left  \{ \frac{2a^2+b^2}{2a}, \frac{2b^2+a^2}{2b} \right \}.
\]
\end{proof}
\begin{rmk}
We now show that these inequalities have nontrivial solutions. 
Suppose $a > b$. Then $$2(a^3 - b^3) + 4ba (b - a) = 2(a - b)(a^2 + ab + b^2) - 4ab(a - b) = 2(a - b)(a - b)^2 > 0.$$
Thus $$
2a(a^2 + 2b^2) > 2b(b^2 + 2a^2)
$$
and hence 
$$
\frac{2a^2+b^2}{2a} <  \frac{2b^2+a^2}{2b}
$$
Thus we are now solving
$$
a < \frac{E}{\mu} < \frac{2a^2+b^2}{2a}
$$
Which has nontrivial solutions for $b > 0$.
\end{rmk}

\section{Acknowledgement} We would like to thank Eugene Lerman and Yuliy Baryshnikov for several helpful discussions.



\end{document}